\documentclass[journal,10pt]{IEEEtran}
\usepackage{amsmath}

\usepackage{graphicx}
\usepackage{array}
\usepackage{mdwmath}
\usepackage{amssymb}
\usepackage{mdwtab}
\usepackage{amsmath,amsthm}
\usepackage{threeparttable}
\usepackage{color}
\usepackage{bm}
\usepackage{url}
\usepackage{algpseudocode}
\usepackage{algorithm}
\usepackage{multicol}
\usepackage{epstopdf}
\usepackage[colorlinks]{hyperref}
\usepackage{tabularx}
\usepackage{booktabs}

\usepackage[caption=false,font=footnotesize]{subfig}

\algrenewcommand{\algorithmicrequire}{\textbf{Input:}}
\algrenewcommand{\algorithmicensure}{\textbf{Output:}}

\newtheorem{remark}{Remark}
\newtheorem{lemma}{Lemma}

\newtheorem{definition}{Definition}

\begin{document}

\title{Spatial Modulation for Tx-SIMO-FAS:\\Port Selection and Performance Analysis}

\author{Xusheng Zhu,
Kai-Kit Wong,~\IEEEmembership{Fellow,~IEEE},
Hanjiang Hong,
Chenguang Rao, and
Kaitao Meng
\vspace{-6mm}

\thanks{The work of C. Rao, K. K. Wong, and H. Hong is supported by the Engineering and Physical Sciences Research Council (EPSRC) under grant EP/W026813/1.}

\thanks{X. Zhu, K. K. Wong, H. Hong, and C. Rao are with the Department of Electronic and Electrical Engineering, University College London, WC1E 7JE London, United Kingdom. K. K. Wong is also affiliated with the Department of Electronic Engineering, Kyung Hee University, Yongin-si, Gyeonggi-do 17104, Republic of Korea (e-mail: \{xusheng.zhu; kai-kit.wong; hanjiang.hong; chenguang.rao\}@ucl.ac.uk).}
\thanks{K. Meng is with the Department of Electrical and Electronic Engineering, University of Manchester, Manchester, United Kingdom (e-mail: kaitao.meng@manchester.ac.uk).}
\thanks{Corresponding author: Kai-Kit Wong.}
}

\maketitle

\begin{abstract}
This paper considers a single-input multiple-output (SIMO) setup with a fluid antenna system (FAS) at the transmitter side and multiple fixed antennas at the receiver, which is referred to as a Tx-SIMO-FAS. We investigate the use of spatial modulation (SM) utilizing the FAS on a single radio-frequency (RF) chain while the receiver side performs maximum-likelihood detection. Unlike conventional antenna arrays, however, the large number of fluid antenna ports accommodated within a limited aperture introduces strong spatial correlation, which reduces the distinguishability of port indices and degrades the reliability of index detection. To address this challenge, three correlation-aware port-selection schemes are proposed: successive fluid Euclidean-distance-optimized selection (SF-EDAS), successive orthogonal port selection (SOPS), and correlation-constrained orthogonal array selection (CC-COAS). These schemes focus on enhancing received-constellation separation, improving channel-basis conditioning, and jointly optimizing channel gain and inter-port decorrelation, respectively. To understand the performance limits of FAS-SM, a reliability analysis is developed by decomposing the channel into an energy-based degree of freedom (DoF), and an extreme-value DoF. High signal-to-noise ratio (SNR) analysis reveals an effective diversity order determined by the number of selected ports, the number of receive antennas, and the energy-based spatial DoF. Furthermore, the aperture-limited array gain is characterized through a scalar equivalent independent-look approximation involving the Digamma function. Numerical results demonstrate that the proposed schemes significantly outperform conventional SM and grouping-based benchmarks. Among them, CC-COAS achieves the most favorable tradeoff between error performance and computational complexity. The results further reveal a fundamental aperture-limited behavior that the array gain eventually saturates as the fluid-port density increases.
\end{abstract}


\vspace{-2mm}
\section{Introduction}
\IEEEPARstart{T}{he maturation} of multiple-input multiple-output (MIMO) technology has greatly advanced wireless communications by exploiting spatial degrees of freedom (DoFs) to improve both capacity and reliability~\cite{alb2019ma}. In recent years, reconfigurable intelligent surfaces (RIS) have also been proposed to perform passive MIMO beamforming and spatial modulation for spatial multiplexing and diversity gains simultaneously~\cite{zhu2024dssm,zhu2023ris}. But conventional MIMO architectures are typically implemented using fixed-position antennas (FPAs), which limits their ability to fully exploit the continuous spatial variation of wireless channels in a given physical aperture.

To overcome this limitation, fluid antenna systems (FAS) have recently emerged as a promising paradigm for integrating reconfigurable antennas into the physical layer for future wireless networks \cite{new2025a,new2026_fas_jsac,hong2026fluid}. By dynamically adjusting physical positions or selecting favorable ports within a prescribed spatial aperture, FAS can exploit channel variations far more flexibly than conventional fixed-position architectures~\cite{wong2020per,wong2021flud,zhe2024mov}. Fundamental studies have characterized the extreme-value performance limits of FAS and revealed the associated diversity-multiplexing tradeoffs over spatial correlated channels~\cite{zhu2025fas,new2024an}. Building on these theoretical foundations, FAS has been extended to a variety of communication paradigms, including fluid antenna multiple access (FAMA) for massive connectivity \cite{rao2026code,hong2025mul}, unmanned aerial vehicle (UAV) communications for short-packet transmission~\cite{zhu2025faev,zhuug2025,zhu2026uav}, and fluid RISs for element-level pattern reconfigurability \cite{xiao2026f} and secure wireless transmission \cite{zhu2026fluid}. These developments highlight the potential of continuous spatial agility for future wireless systems.

In parallel, index modulation (IM) is as an efficient transmission framework that conveys additional information bits through the activation states of communication resources~\cite{wen2019a,zhu2026ind}. Amongst its various forms, spatial modulation (SM) is particularly attractive since it uses transmit-antenna indices to carry information while activating only one radio-frequency (RF) chain at a time~\cite{zhu2022on}. This mechanism enables SM to achieve a favorable tradeoff between spectral efficiency and hardware complexity, while avoiding inter-channel interference and reducing transceiver complexity~\cite{li2021single}. SM has already been extended to a range of emerging architectures, such as spatial scattering modulation in millimeter-wave (mmWave) MIMO systems~\cite{zhu2025spat,zhu2025toward} and RIS-assisted SM schemes based on transmissive and reflective surfaces~\cite{zhu2025trans,zhu2024on}. These works demonstrate the versatility of SM as a low-complexity and energy-efficient modulation principle.

Motivated by the complementary strengths of FAS and SM, there is a growing interest in their joint integration. On the one hand, FAS provides continuous spatial agility and rich port selectability within a given aperture. On the other hand, SM offers a natural mechanism to encode information through port indices on a single active RF chain. Early studies on FAS-assisted SM mainly focused on single-antenna reception scenarios in order to establish basic port-to-bit mapping principles. For example, the port-grouping method in~\cite{yang2024position} partitioned the fluid-port array into several physical groups and activated one port per group according to a heuristic distance criterion, while the FAS-SM orthogonal frequency-division multiplexing (OFDM) scheme in~\cite{chen2024fas} mapped information bits to dynamic port positions across different subcarriers. However, such single-antenna reception models inherently limit the achievable spatial diversity and overall system capacity.

To improve data rates, subsequent studies considered multi-antenna FAS-IM architectures with richer spatial activation strategies. Specifically, in~\cite{zhu2024index}, multiple RF chains were employed to map data streams onto different port combinations through field-response channels. In~\cite{guo2025fluid}, grouped IM further enhanced spectral efficiency by dividing the spatial constellation into several subgroups and simultaneously activating one port within each subgroup. Although these architectures improve the transmission rate, they often move away from the single-active-port principle of conventional SM by using multiple activated ports or port combinations. Such designs may require multiple RF chains or a more complex RF feeding/control network when the activated ports carry independent data streams or require independent amplitude/phase control. Therefore, the low hardware cost and energy efficiency of conventional single-RF-chain SM may be weakened. To further enhance reliability, coded IM was introduced in~\cite{fadd2025advan}, where spatial set-partition coding and turbo coding were jointly used to improve error performance. Related extensions have also been reported for RIS-aided millimeter-wave channels~\cite{zhu2025fluid}, receive-side IM~\cite{guo2025far}, secure communications~\cite{liu2026index}, FAS-assisted rectangular differential IM~\cite{zhang2026flu}, and other architectures~\cite{wang2026moa}.


Despite these advances, the physical-layer design of FAS-enabled SM under strong spatial correlation remains insufficiently understood. In a compact FAS aperture, a large number of candidate ports sample highly correlated channel responses. Therefore, different port indices may generate similar received signal vectors, which reduces the minimum Euclidean distance of the joint spatial-symbolic constellation. This effect directly degrades index detection and may offset the selection gain expected from dense port deployment. Existing studies have mainly focused on architectural extensions, grouping heuristics, or coding-aided designs. However, a systematic constellation design framework that jointly accounts for channel gain, spatial separability, computational complexity, and the single-RF-chain constraint is missing, which motivates this work. The main contributions of this paper are summarized as follows:
\begin{enumerate}
\item We propose a FAS-based SM system, referred to as the Tx-single-input multiple-output (SIMO)-FAS architecture according to \cite[Table IV]{new2025a}, in which the transmitter activates one selected port through a single RF chain and the receiver employs a conventional fixed multi-antenna array. The information bits are jointly conveyed by the activated port index and the modulation symbol. A maximum-likelihood (ML) detector is used to recover both the spatial index and the transmitted symbol.
\item We formulate the spatial-domain constellation design problem caused by strong inter-port correlation in FAS apertures. To address this problem with practical complexity, we develop three port-selection algorithms.  First, successive fluid Euclidean-distance-optimized selection (SF-EDAS) improves the minimum Euclidean distance of the received constellation. Second, successive orthogonal port selection (SOPS) constructs a nearly orthogonal channel basis by successive null-space projection. Third, correlation-constrained orthogonal array selection (CC-COAS) uses an adaptive correlation constraint to balance channel-gain and spatial decorrelation.
\item Moreover, we develop an analytical reliability framework for the proposed Tx-SIMO-FAS SM architecture. The framework includes an algorithm-aware pairwise error probability (PEP) expression, a closed-form static baseline, and a high-signal-to-noise ratio (SNR) asymptotic characterization. To clarify the role of the physical aperture, we separate the spatial channel into an energy-based DoF for diversity analysis and an extreme-value DoF for array-gain analysis. The resulting expression reveals an order-statistics penalty introduced by the need to select a multi-port subset for SM.
\item We provide extensive simulation results to support the theoretical analysis and compare the proposed algorithms with conventional SM and grouping-based benchmarks. The results show that SF-EDAS and CC-COAS provide strong reliability gains, while SOPS becomes more effective when the aperture is sufficiently large. Among the three methods, CC-COAS achieves the best overall tradeoff between reliability and complexity. The results further confirm that increasing the port density improves performance but only up to an aperture-limited saturation level, consistent with the existing literature.
\end{enumerate}

The remainder of this paper is organized as follows. In Section~\ref{sec:system_model}, we introduce the FAS-SM system model and the spatial correlation model. Section~\ref{sec:constellation_design} presents the proposed correlation-aware constellation design algorithms. Section~\ref{sec:performance_analysis} develops the reliability analysis and the asymptotic performance characterization. Section~\ref{sec:numerical_results} provides numerical results and discussions. Section~\ref{sec:conclusion} concludes the paper.

\section{System Model} \label{sec:system_model}
Consider a point-to-point MIMO link over a quasi-static Rayleigh block-fading channel. To strike a desirable balance between hardware cost and system performance, the transmitter incorporates a FAS driven by a single active RF chain, while the receiver is equipped with $N_{\rm r}$ conventional FPAs. This model is referred to as Tx-SIMO-FAS \cite[Table IV]{new2025a}. The single RF chain can be dynamically connected to one of $N$ electronically accessible fluid ports. These candidate ports are uniformly distributed along a one-dimensional linear aperture of length $W\lambda$, where $W$ and $\lambda$ denote the normalized aperture size and the carrier wavelength, respectively. Accordingly, the coordinate of the $n$-th port is expressed as
\begin{equation} \label{eq:port_location}
z_n = \frac{n-1}{N-1}W\lambda,~n=1,2,\ldots,N.
\end{equation}

\begin{figure}[]
\centering
\includegraphics[width=.95\linewidth]{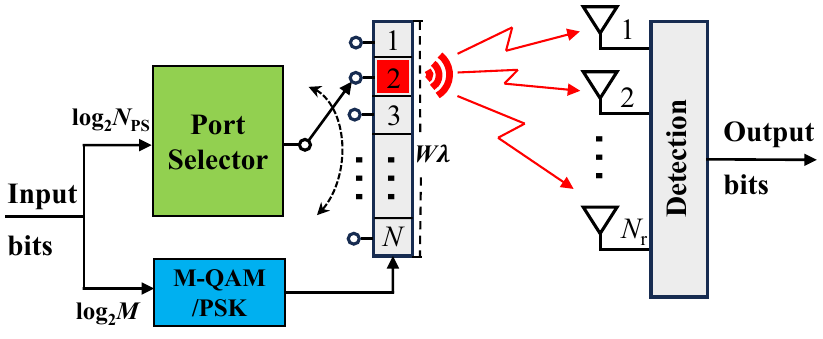}
\caption{System model of the proposed Tx-SIMO-FAS SM scheme.}\label{fig:sys_model}
\end{figure}

\subsection{Signal Transmission}
Fig.~\ref{fig:sys_model} shows the proposed SM transmission scheme under the Tx-SIMO-FAS model. Before data transmission, the transmitter selects $N_{\rm PS}$ ports from the full set of $N$ candidate ports. The selected ports form the spatial alphabet used for SM. Let $\mathcal{I}_a \subset \{1,2,\ldots,N\}$ denote the selected port set, where $|\mathcal{I}_a|=N_{\rm PS}$. Due to the single-RF-chain constraint, only one port in $\mathcal{I}_a$ is activated in each channel use.

During each transmission block, the incoming data stream is partitioned into groups of $B$ bits per channel use (bpcu), which is formulated as
\begin{equation}
B = \log_{2}(N_{\rm PS}) + \log_{2}(M).
\end{equation}
Upon splitting the bit sequence, the data is mapped into two strictly orthogonal signal domains. First, the primary $\log_{2}(M)$ bits are fed into the modulation mapper to generate a complex symbol $s$ drawn from a unit-energy constellation alphabet $\mathcal{S}$, satisfying the power constraint $\mathbb{E}[|s|^{2}] = 1$. Secondly, the residual $\log_{2}(N_{\rm PS})$ bits are utilized to drive the functional port selector. This selector dynamically determines the spatial index $i \in \mathcal{I}_{a}$ and switches the RF chain, dictating that the generated symbol $s$ will be physically radiated exclusively from the $i$-th \textit{selected} fluid port ($i \in \mathcal{I}_a$). This mechanism ensures that although $N_{\rm PS}$ ports are available for indexing, the single RF chain only illuminates one specific location at any given instant, thereby maintaining low hardware complexity.

The transmit vector is denoted by $\mathbf{x} \in \mathbb{C}^{N \times 1}$. If the $i$-th port is activated to transmit symbol $s$, the vector is formulated as
$\mathbf{x} = \mathbf{e}_{i} s$,
where $\mathbf{e}_{i}$ denotes the $i$-th column of the $N \times N$ identity matrix $\mathbf{I}_{N}$. Consequently, the system operates over a joint spatial-symbolic alphabet defined as
\begin{equation}
    \mathcal{X} = \{\mathbf{e}_{i} s \mid i \in \mathcal{I}_{a}, s \in \mathcal{S}\}, \label{eq:joint_alphabet}
\end{equation}
with a total cardinality of $|\mathcal{X}| = N_{\rm PS}M$.

The received signal vector $\mathbf y\in\mathbb C^{N_r\times 1}$ is given by
\begin{equation}
\mathbf{y} = \sqrt{\gamma}\mathbf{H}\mathbf{x}+\mathbf{n} = \sqrt{\gamma}\mathbf{h}_i s+\mathbf{n}, \label{eq:received_signal}
\end{equation}
in which $\gamma$ is the average transmit SNR per receive antenna, $\mathbf{H}=[\mathbf{h}_1,\ldots,\mathbf{h}_N]\in\mathbb{C}^{N_r\times N}$ denotes the spatially correlated channel matrix, $\mathbf{h}_i$ is the channel vector associated with port $i$, and the noise vector is modeled as $\mathbf{n}\sim\mathcal{CN}(\mathbf{0},\mathbf{I}_{N_r})$.

\subsection{Correlation Model}
Since the ports are closely spaced within a compact aperture, their channel responses are spatially correlated. Under a two-dimensional isotropic scattering model, the transmit-side correlation between ports $m$ and $n$ follows Jakes' model as
\begin{equation}\label{eq:jakes_correlation}
R_{m,n} = 
J_0\left(\frac{2\pi W |m-n|}{N-1}\right),
\end{equation}
where $J_0(\cdot)$ is the zeroth-order Bessel function of the first kind. The receive FPAs are assumed to be sufficiently separated, so the receive-side spatial correlation is ignored. Thus, the channel matrix is modeled as
\begin{equation}\label{eq:kronecker_model}
\mathbf{H} = \mathbf{H}_{w} \mathbf{R}^{1/2},
\end{equation}
where $\mathbf{H}_w\in\mathbb{C}^{N_r\times N}$ has independent $\mathcal{CN}(0,1)$ entries, and $\mathbf{R}\in\mathbb{R}^{N\times N}$ represents the transmit-side correlation matrix with $[\mathbf{R}]_{m,n}=R_{m,n}$, given in (\ref{eq:jakes_correlation}) above.

\subsection{Signal Detection}
Given the selected port set $\mathcal{I}_a=\{\pi_1,\pi_2,\ldots,\pi_{N_{\rm PS}}\}$, the receiver jointly detects the active port index and the modulation symbol. With perfect channel state information at the receiver, ML detection can be expressed as
\begin{equation}\label{eq:ml_detection}
(\hat{i}, \hat{s}) = \arg \min_{i \in \mathcal{I}_a, s \in \mathcal{S}} \|\mathbf{y} - \sqrt{\gamma} \mathbf{h}_i s\|_2^2,
\end{equation}
where $\hat{i}$ and $\hat{s}$ denote the estimated port index and modulation symbol, respectively.

The reliability of \eqref{eq:ml_detection} depends on the geometry of the effective channel submatrix $\mathbf{H}_{\mathcal{I}_a}=[\mathbf{h}_{\pi_1},\ldots,\mathbf{h}_{\pi_{N_{\rm PS}}}]$. If the selected ports are strongly correlated, the corresponding channel vectors become nearly aligned. Different spatial indices may then produce similar received signal vectors, which reduces the minimum distance of the joint constellation and increases the probability of index errors. A reliable port set should therefore provide both sufficient channel gain and sufficient spatial separation. The next section develops three correlation-aware selection methods for this purpose.

\section{Correlation-Aware Constellation Design} \label{sec:constellation_design}
This section develops a spatial-symbolic constellation design framework to address both the collinearity bottleneck and the computational intractability of dense FAS architectures. To address these challenges, we develop three low-complexity subset selection paradigms from geometric, algebraic, and physical perspectives, respectively. By systematically balancing the tradeoff between array gain and spatial separability, these algorithms can greatly improve reliability.

\begin{algorithm}[t]
\caption{Proposed SF-EDAS algorithm}
\label{alg:sf_edas}
\begin{algorithmic}[1]
\Require Channel matrix $\mathbf{H} \in \mathbb{C}^{N_{\rm r} \times N}$, subset size $N_{\rm PS}$, square quadrature amplitude modulation (QAM) alphabet $\mathcal{S}$
\Ensure Optimized active spatial subset $\mathcal{I}_{\text{SF}}$

\State Initialize pairwise distance matrix $\mathbf{D} \leftarrow \mathbf{0}_{N \times N}$

\For{$i = 1$ \textbf{to} $N-1$}
    \For{$j = i+1$ \textbf{to} $N$}
        \State Calculate the channel phase rotation via \eqref{eq:phase_rotation}
        \State Compute the pairwise distance $D_{i,j}$ via \eqref{eq:decoupled_metric}
        \State $D_{j,i} \leftarrow D_{i,j}$ 
    \EndFor
\EndFor

\State Initialize candidate pool $\mathcal{P} \leftarrow \{1, \dots, N\}$
\State $i_1 \leftarrow \arg \max_{n \in \mathcal{P}} \|\mathbf{h}_n\|_2^2$
\State $\mathcal{I}_{\text{SF}} \leftarrow \{i_1\}$, $\mathcal{P} \leftarrow \mathcal{P} \setminus \{i_1\}$

\For{$k = 2$ \textbf{to} $N_{\rm PS}$}
    \State Select $i_k = \arg \max_{n \in \mathcal{P}} \{ \min_{m \in \mathcal{I}_{\text{SF}}} D_{n,m} \}$ via \eqref{eq:sf_edas_objective_math}
    \State $\mathcal{I}_{\text{SF}} \leftarrow \mathcal{I}_{\text{SF}} \cup \{i_k\}$, $\mathcal{P} \leftarrow \mathcal{P} \setminus \{i_k\}$
\EndFor

\State \Return $\mathcal{I}_{\text{SF}}$
\end{algorithmic}
\end{algorithm}

\subsection{Proposed SF-EDAS Algorithm}
The first design directly uses the Euclidean distance of the received constellation as the selection metric. For any two distinct spatial-symbolic vectors $\mathbf{x}_{i,q}=\mathbf{e}_i s_q$ and $\mathbf{x}_{j,p}=\mathbf{e}_j s_p$, the squared Euclidean distance in the noise-free received signal space can be expressed as
\begin{equation}
    d_{\rm E}^2(\mathbf{x}_{i,q},\mathbf{x}_{j,p})
    =
    \|\mathbf{H}(\mathbf{x}_{i,q}-\mathbf{x}_{j,p})\|_2^2
    =
    \mathbf{s}^H\mathbf{G}_{i,j}\mathbf{s},
    \label{eq:quadratic_form_distance}
\end{equation}
where $\mathbf{s}=[s_q,-s_p]^T$, and
\begin{equation}
    \mathbf{G}_{i,j}
    =
    \begin{bmatrix}
    \|\mathbf{h}_i\|_2^2 & \mathbf{h}_i^H\mathbf{h}_j\\
    \mathbf{h}_j^H\mathbf{h}_i & \|\mathbf{h}_j\|_2^2
    \end{bmatrix}.
    \label{eq:sf_edas_gram}
\end{equation}
After some simple calculations, \eqref{eq:quadratic_form_distance} can be rewritten as
\begin{equation}
    d_{\rm E}^2 = |s_q|^2 \|\mathbf{h}_i\|_2^2 + |s_p|^2 \|\mathbf{h}_j\|_2^2 - 2 \operatorname{Re} \{ s_q s_p^* \mathbf{h}_j^H \mathbf{h}_i \}, \label{eq:sf_edas_expanded_distance}
\end{equation}
where the distance depends on both the channel norms and the inner product between two port channels. Highly correlated ports can reduce the distance between spatial-symbolic states, especially when the symbol phases align.

A direct search over all port subsets is not practical when $N$ is large. SF-EDAS therefore first builds a pairwise distance matrix $\mathbf{D}\in\mathbb{R}^{N\times N}$ and then selects ports by successive expansion, as detailed in Algorithm \ref{alg:sf_edas}.
For each port pair $(i,j)$, the channel inner product can be written as
\begin{equation}\label{eq:phase_rotation}
\mathbf{h}_{j}^{H}\mathbf{h}_{i}=\left|\mathbf{h}_{j}^{H}\mathbf{h}_{i}\right|e^{\mathrm{j}\phi_{i,j}},
\end{equation}
where $\phi_{i,j}$ represents the phase of the channel inner product. This phase term identifies the relative symbol alignment that can
reduce the pairwise Euclidean distance. For square QAM, let $\mathcal{A}$ denote the corresponding one-dimensional pulse amplitude modulation (PAM) alphabet, i.e., $\mathcal{S}=\left\{s_I+\mathrm{j}s_Q\mid s_I,s_Q\in\mathcal{A}\right\},  |\mathcal{A}|=\sqrt{M}$. To avoid an exhaustive search over all $M^2$ symbol pairs for each candidate port pair, SF-EDAS adopts a phase-aligned separable metric. Based on \eqref{eq:sf_edas_expanded_distance}, the pairwise distance metric is defined as
\begin{equation}\label{eq:decoupled_metric}
\begin{aligned}
    D_{i,j}
    \triangleq
    &\min_{a_I,b_I\in\mathcal{A}}
    \Big(
    a_I^{2}\|\mathbf{h}_i\|_2^{2}
    + b_I^{2}\|\mathbf{h}_j\|_2^{2}
    -2a_I b_I
    \left|\mathbf{h}_{j}^{H}\mathbf{h}_{i}\right|
    \Big)  \\
    &+
    \min_{a_Q,b_Q\in\mathcal{A}}
    \Big(
    a_Q^{2}\|\mathbf{h}_i\|_2^{2}
    + b_Q^{2}\|\mathbf{h}_j\|_2^{2}
    -2a_Q b_Q
    \left|\mathbf{h}_{j}^{H}\mathbf{h}_{i}\right|
    \Big).
\end{aligned}
\end{equation}
The two minimizations are performed over the in-phase and quadrature PAM components, respectively.

After the matrix $\mathbf{D}$ is obtained, SF-EDAS selects the active subset in a greedy max-min manner. The subset is initialized with the port having the largest channel norm. If $k-1$ ports have been selected, the $k$-th port is chosen as
\begin{equation}\label{eq:sf_edas_objective_math}
i_k = \arg \max_{n \notin \mathcal{I}_{\text{SF}}^{(k-1)}} \min_{m \in \mathcal{I}_{\text{SF}}^{(k-1)}} D_{n,m}.
\end{equation}
The resulting subset, from this approach, tends to keep the received constellation well separated.

\begin{remark}
SF-EDAS improves the minimum distance of the joint spatial-symbolic constellation and is therefore effective in reducing index ambiguity caused by correlated ports. However, its metric is dominated by the worst-case pairwise distance. As a result, some high-gain ports may be discarded if their received signal vectors are too close to those of the already selected ports. This leads to a tradeoff between distance maximization and channel-gain selection.
\end{remark}

\begin{algorithm}[t]
\caption{Proposed SOPS algorithm}
\label{alg:sops}
\begin{algorithmic}[1]
\Require Channel matrix $\mathbf{H} \in \mathbb{C}^{N_{\rm r} \times N}$, subset size $N_{\rm PS}$
\Ensure Decoupled active spatial subset $\mathcal{I}_{\text{SOPS}}$

\State Initialize candidate pool $\mathcal{P} \leftarrow \{1, \dots, N\}$
\State $i_1 \leftarrow \arg \max_{n \in \mathcal{P}} \|\mathbf{h}_n\|_2^2$ 
\State $\mathcal{I}_{\text{SOPS}} \leftarrow \{i_1\}$, $\mathcal{P} \leftarrow \mathcal{P} \setminus \{i_1\}$

\State Initialize  $\mathbf{P}_{\bot} \leftarrow \mathbf{I}_{N_r} - \displaystyle \frac{\mathbf{h}_{i_1} \mathbf{h}_{i_1}^H}{\|\mathbf{h}_{i_1}\|_2^2}$

\For{$k = 2$ \textbf{to} $N_{\rm PS}$}
    \State $v_n \leftarrow \| \mathbf{P}_{\bot} \mathbf{h}_n \|_2^2, \quad \forall n \in \mathcal{P}$
    \State $i_k \leftarrow \arg \max_{n \in \mathcal{P}} v_n$
    \State $\mathcal{I}_{\text{SOPS}} \leftarrow \mathcal{I}_{\text{SOPS}} \cup \{i_k\}$, $\mathcal{P} \leftarrow \mathcal{P} \setminus \{i_k\}$

    \If{$v_{i_k} > \epsilon$} 
        \State Update $\mathbf{P}_{\bot} \leftarrow \mathbf{P}_{\bot} - \displaystyle \frac{\mathbf{P}_{\bot} \mathbf{h}_{i_k} \mathbf{h}_{i_k}^H \mathbf{P}_{\bot}}{v_{i_k}}$ 
    \EndIf
\EndFor

\State \Return $\mathcal{I}_{\text{SOPS}}$
\end{algorithmic}
\end{algorithm}

\subsection{Proposed SOPS Algorithm}
Although the SF-EDAS technique improves the Euclidean distance, it requires a pairwise distance matrix whose computation scales quadratically with $N$. To reduce the dependence on symbol-level distance search, SOPS in Algorithm \ref{alg:sops} selects ports according to their contribution outside the subspace spanned by the already selected ports.

Specifically, the first selected port is chosen as the strongest channel vector:
\begin{equation}\label{eq:sops_init}
i_1 = \arg \max_{n \in \{1, \dots, N\}} \|\mathbf{h}_n\|_2^2.
\end{equation}
Let $\mathcal{V}_{k-1}$ denote the subspace spanned by the first $k-1$ selected channel vectors. The projection matrix onto the orthogonal complement of this subspace is given by
\begin{equation}\label{eq:sops_proj}
\mathbf{P}_{\bot}^{(k-1)} = \mathbf{I}_{N_r} - \mathbf{H}^{(k-1)} ( (\mathbf{H}^{(k-1)})^H \mathbf{H}^{(k-1)} )^{-1} (\mathbf{H}^{(k-1)})^H,
\end{equation}
where $\mathbf{H}^{(k-1)}$ contains the previously selected channel vectors. For an unselected port $n$, the projection energy is
\begin{equation}\label{eq:sops_innov}
E_{\mathrm{innov}}^{(k-1)}(n) = \| \mathbf{P}_{\bot}^{(k-1)} \mathbf{h}_n \|_2^2 = \mathbf{h}_n^H \mathbf{P}_{\bot}^{(k-1)} \mathbf{h}_n.
\end{equation}
SOPS then selects the port with the largest projection energy. To avoid repeated matrix inversion, the projection matrix can be updated recursively as
\begin{equation}\label{eq:sops_greville}
\mathbf{P}_{\bot}^{(k)} = \mathbf{P}_{\bot}^{(k-1)} - \frac{\mathbf{P}_{\bot}^{(k-1)} \mathbf{h}_{i_k} \mathbf{h}_{i_k}^H \mathbf{P}_{\bot}^{(k-1)}}{\mathbf{h}_{i_k}^H \mathbf{P}_{\bot}^{(k-1)} \mathbf{h}_{i_k}}.
\end{equation}

\begin{remark}
SOPS improves the conditioning of the selected channel matrix by choosing ports that add large projection energy outside the current selected subspace. The method is independent of the modulation alphabet and avoids the pairwise symbol search required by SF-EDAS. However, the projection metric does not directly preserve the original channel norm. In compact apertures, SOPS may select low-gain ports to satisfy the orthogonality requirement, which can lead to a power loss.
\end{remark}

\begin{algorithm}[t]
\caption{Proposed CC-COAS algorithm}
\label{alg:cc_coas}
\begin{algorithmic}[1]
\Require Channel matrix $\mathbf{H} \in \mathbb{C}^{N_{\rm r} \times N}$, subset size $N_{\rm PS}$, threshold candidates $\mathcal{R} = \{\rho_1, \dots, \rho_{L}\}$ where $\mathcal R\subset[0,1)$.
\Ensure Optimal active spatial subset $\mathcal{I}_{\text{opt}}$

\State $E_n \leftarrow \|\mathbf{h}_n\|_2, \quad \forall n \in \{1, \dots, N\}$
\State Obtain sorted index vector $E_{\pi_1} \ge E_{\pi_2} \ge \dots \ge E_{\pi_N}$

\State Initialize $M_{\text{opt}} \leftarrow -\infty$ and $\mathcal{I}_{\text{opt}} \leftarrow \emptyset$

\For{each threshold $\rho_l \in \mathcal{R}$}
    \State Initialize temporary subset $\mathcal{I}_{\text{tmp}} \leftarrow \{\pi_1\}$, $idx \leftarrow 2$

    \While{$|\mathcal{I}_{\text{tmp}}| < N_{\rm PS}$ \textbf{and} $idx \le N$}
        \State $k \leftarrow \pi_{idx}$
        \State $\phi_{\max} \leftarrow \max\limits_{m \in \mathcal{I}_{\text{tmp}}} \displaystyle \frac{|\mathbf{h}_k^H \mathbf{h}_m|}{\|\mathbf{h}_k\|_2 \|\mathbf{h}_m\|_2}$

        \If{$\phi_{\max} \le \rho_l$}
            \State $\mathcal{I}_{\text{tmp}} \leftarrow \mathcal{I}_{\text{tmp}} \cup \{k\}$
        \EndIf
        \State $idx \leftarrow idx + 1$
    \EndWhile

    \If{$|\mathcal{I}_{\text{tmp}}| = N_{\rm PS}$}
       \State Compute $M_{\rm tmp}=C(\mathcal I_{\rm tmp})$ using \eqref{eq:capacity_fas}
        \If{$M_{\text{tmp}} > M_{\text{opt}}$}
            \State $M_{\text{opt}} \leftarrow $ $M_{\text{tmp}}$
            \State $\mathcal{I}_{\text{opt}} \leftarrow \mathcal{I}_{\text{tmp}}$
        \EndIf
    \EndIf
\EndFor

\State \Return $\mathcal{I}_{\text{opt}}$
\end{algorithmic}
\end{algorithm}

\begin{figure*}[t]
\centering
\includegraphics[width=0.95\linewidth]{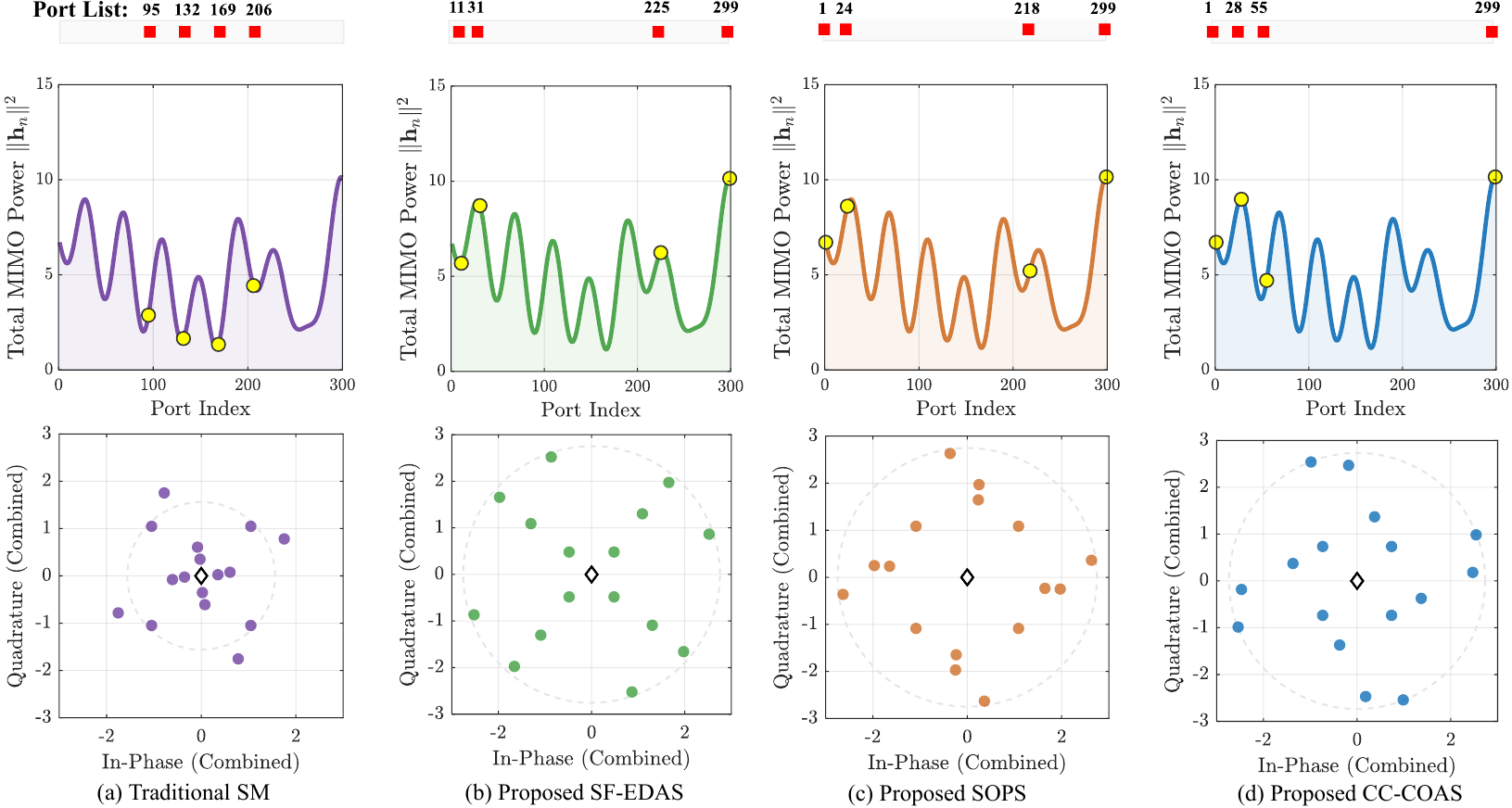}
\caption{Mechanism and decision geometry analysis under different design criteria ($W=4\lambda, N_{\rm PS}=M=N_{\rm r}=4$).}\label{fig:mechanism}
\end{figure*}

\subsection{Proposed CC-COAS Algorithm}
CC-COAS is designed to balance channel-gain selection and spatial-rank preservation. For a selected port set $\mathcal{I}_a$, the equivalent MIMO capacity can be computed as
\begin{equation}\label{eq:capacity_fas}
C(\mathcal{I}_a) = \log_{2} \det \left( \mathbf{I}_{N_r} + \frac{\gamma}{N_{\rm PS}} \mathbf{H}_{\mathcal{I}_a} \mathbf{H}_{\mathcal{I}_a}^H \right).
\end{equation}
For fixed column norms, the Hadamard inequality shows that a better-conditioned channel matrix is obtained when the selected columns are close to orthogonal. In a FAS aperture, however, strict orthogonality may force the selection of weak ports. CC-COAS therefore uses a gain-aware correlation constraint rather than a pure orthogonality rule.

The normalized channel correlation between ports $i$ and $j$ is defined as
\begin{equation}\label{eq:cc_coas_metric}
\phi(i,j) = \frac{|\mathbf{h}_{i}^{H} \mathbf{h}_{j}|}{\|\mathbf{h}_{i}\|_{2} \|\mathbf{h}_{j}\|_{2}}.
\end{equation}
Let $\{\pi_1,\ldots,\pi_N\}$ denote the port indices sorted in descending order of $\|\mathbf{h}_n\|_2$. For a given threshold $\rho_l$, CC-COAS scans this ordered list and adds candidate $\pi_n$ to the temporary subset $\mathcal{I}_{\mathrm{tmp}}$ if
\begin{equation}\label{eq:cc_coas_constraint}
\max_{m \in \mathcal{I}_{\text{tmp}}} \phi(\pi_{n}, m) \le \rho_l.
\end{equation}
The threshold $\rho_l$ is a parameter that controls the balance between channel gain and spatial separation. A smaller threshold enforces stronger decorrelation, while a larger threshold allows more high-gain ports to be retained.

\begin{lemma}\label{lemma:rank_preservation}
Suppose that the threshold satisfies $0\le \rho_l<1$. Under the constraint in \eqref{eq:cc_coas_constraint}, any selected port pair remains non-collinear, and the corresponding pairwise Gram matrix $\mathbf C_{i,j}$ is positive definite. Equivalently,
\begin{equation}\label{eq:lambda_min_positivity}
\lambda_{\min}(\mathbf{C}_{i,j})>0.
\end{equation}
\end{lemma}

\begin{proof}
See Appendix~\ref{appendix:proof_lemma}.
\end{proof}

\begin{remark}
A fixed correlation threshold gives a simple selection rule, but the best balance between channel gain and spatial separation depends on the instantaneous channel. As a consequence, CC-COAS searches over a finite set of threshold candidates $\mathcal{R}$. For each feasible threshold, the algorithm forms a candidate subset and evaluates the corresponding performance metric. The subset with the best metric is selected as $\mathcal{I}_{\mathrm{opt}}$, as shown in Algorithm~\ref{alg:cc_coas}.
\end{remark}

\subsection{Visual of Port-to-Constellation Mappings}
Fig.~\ref{fig:mechanism} illustrates how different port-selection rules affect the received constellation under the same channel realization. The figure compares four schemes from three views: the selected port locations, the channel-gain envelope, and the resulting received constellation.

\subsubsection{Traditional SM in \textbf{Fig.~\ref{fig:mechanism}(a)}}
Traditional SM uses a fixed and uniformly spaced port set. Since the selected ports do not adapt to the instantaneous channel, some of them may lie in low-gain regions. The corresponding received points are then closer to the origin, which reduces the minimum Euclidean distance and degrades the average error performance.

\subsubsection{SF-EDAS in \textbf{Fig.~\ref{fig:mechanism}(b)}}
SF-EDAS selects ports by improving the worst-case Euclidean distance. This rule helps separate the spatial-index states in the received constellation. However, as the metric is distance oriented, the selected subset may not always include the strongest channel-gain peaks.

\subsubsection{SOPS in \textbf{Fig.~\ref{fig:mechanism}(c)}}
SOPS selects ports according to their projection energy outside the subspace spanned by the already selected ports. This rule improves the conditioning of the selected channel matrix and reduces spatial collinearity. But as the aperture is small, the orthogonality requirement may lead to the selection of ports with weak channel gains.

\subsubsection{CC-COAS in \textbf{Fig.~\ref{fig:mechanism}(d)}}
CC-COAS provides a more balanced selection. It first gives priority to high-gain ports and then rejects candidates that are too correlated with the current subset. As a result, the selected ports maintain both sufficient received power and adequate spatial separation. This explains why CC-COAS achieves a favorable tradeoff between reliability and complexity in the numerical results.

\subsection{Complexity Analysis}
In this subsection, the computational complexity is analyzed in terms of the dominant operations per coherence block. The main parameters are the number of candidate ports in the transmit FAS $N$, the number of receive FPAs $N_{\rm r}$, the selected subset size $N_{\rm PS}$, and the modulation order $M$.

\subsubsection{Complexity of SF-EDAS}
SF-EDAS first computes the pairwise distance matrix and then performs greedy max-min selection. The distance-matrix computation requires $\mathcal{O}(N^2MN_{\rm r})$ operations, while the subset expansion over $N_{\rm PS}$ iterations requires $\mathcal{O}(NN_{\rm PS}^2)$ scalar operations. The dominant complexity is therefore
$T_{\mathrm{SF\text{-}EDAS}}=\mathcal{O}(N^2MN_{\rm r}+NN_{\rm PS}^2).$

\subsubsection{Complexity of SOPS}
SOPS avoids symbol-level distance search but requires projection operations. Computing the projection energy and updating the projection matrix require $\mathcal{O}(N_{\rm r}^2)$ operations per candidate in each iteration. The dominant complexity is $T_{\mathrm{SOPS}}=\mathcal{O}(NN_{\rm PS}N_{\rm r}^2).$

\subsubsection{Complexity of CC-COAS}
CC-COAS first sorts the channel norms and then checks pairwise correlations under a finite set of thresholds. Computing and sorting the channel norms require $\mathcal{O}(NN_{\rm r}+N\log N)$ operations. Since the number of threshold candidates is fixed, the dominant complexity is $T_{\mathrm{CC\text{-}COAS}}=\mathcal{O}(N\log N+NN_{\rm PS}N_{\rm r})$.

\section{Performance Analysis} \label{sec:performance_analysis}
This section analyzes the reliability of the proposed Tx-SIMO-FAS SM architecture. We first characterize the effective spatial DoFs imposed by the finite aperture. We then derive an algorithm-aware PEP expression and a closed-form static baseline. Finally, we study the high-SNR behavior to characterize the static diversity order, the predicted effective diversity slope under adaptive port selection, and the aperture-limited extreme-value array gain.

\subsection{Fundamental Spatial DoFs and Decoupled Limits} \label{subsec:spatial_dofs}
Although a FAS may contain a large number of candidate ports, the spatial correlation imposed by a finite aperture limits the number of effectively distinguishable channel observations. In the following, we employ two effective DoF measures to separately characterize the diversity-related energy dispersion and the array-gain-related peak-tracking capability.

\begin{definition} \label{prop:decoupled_dof_final}
The energy-based spatial DoF is defined from the eigenvalues of the transmit-side spatial covariance matrix ${\bf R}$ as
\begin{equation}\label{eq:ds_result_condensed}
D_s = \frac{\mathrm{Tr}(\mathbf{R})^2}{\|\mathbf{R}\|_F^2}.
\end{equation}
The effective extreme-value DoF $N_{\mathrm{EV}}$ describes the number of independent spatial looks available for tracking fading peaks. For a finite number of candidate ports $N$ within a continuous aperture of normalized size $W$, it can be approximated by
\begin{equation}\label{eq:nev_result_condensed}
N_{\mathrm{EV}} \approx \min\{N,\,3.21W+1\}.
\end{equation}
\end{definition}

\begin{proof}
See Appendix \ref{App:Proof_DoF}.
\end{proof}

\begin{remark}
These two DoFs describe different roles of the aperture. The energy-based DoF $D_s$ is related to the diversity gain as it measures the effective rank of the spatial covariance matrix. In contrast, $N_{\mathrm{EV}}$ is related to the array gain because it measures how many nearly independent fading peaks can be observed. Thus, increasing the port density does not create unlimited spatial gain. Once the continuous aperture is sufficiently sampled, the gain is limited by the aperture size.
\end{remark}

\subsection{Reliability Evaluation and Analytical Bounding} \label{subsec:reliability_eval}
\subsubsection{Algorithm-Aware Tight PEP} \label{subsubsec:tight_pep}
Let $\mathbf{x}_{i,q}=\mathbf{e}_i s_q$ denote the transmitted spatial-symbolic vector. Under the ML detector in \eqref{eq:ml_detection}, the instantaneous PEP from $\mathbf{x}_{i,q}$ to $\mathbf{x}_{j,p}$ is found as
\begin{equation}\label{eq:conditional_q_function}
\begin{aligned}
&P(\mathbf{x}_{i,q}\to\mathbf{x}_{j,p}\mid \mathbf{H})\\ &= \Pr\left( \|\mathbf{y}-\sqrt{\gamma}\mathbf{h}_j s_p\|_2^2 < \|\mathbf{y}-\sqrt{\gamma}\mathbf{h}_i s_q\|_2^2 \,\middle|\,\mathbf{H} \right) \nonumber\\ &= Q\left( \sqrt{ \frac{\gamma}{2} \|\mathbf{h}_i s_q-\mathbf{h}_j s_p\|_2^2 } \right).
\end{aligned}
\end{equation}
For a given port-selection algorithm, the selected ports depend on the instantaneous channel realization. The average PEP should therefore be taken over the channel statistics induced by the actual selected subset:
\begin{equation}\label{eq:tight_pep}
P_{\text{PEP}}^{\rm tight}(\mathbf{x}_{i,q} \to \mathbf{x}_{j,p}) = \mathbb{E}_{\mathbf{H}} \left[ Q \left( \sqrt{ \frac{\gamma}{2} \|\mathbf{h}_i s_q - \mathbf{h}_j s_p\|_2^2 } \right) \right],
\end{equation}
which captures the selection gain obtained by the proposed algorithms. Deriving an exact closed-form expression is mathematically intractable since the selected channel vectors are highly non-linear functions of the instantaneous fading realizations under the proposed port-selection rules. To isolate the spatial selection gain, we derive a tractable static baseline.

\subsubsection{Static Baseline for Spatial Correlation Penalty} \label{subsubsec:loose_pep}
We define a static baseline to isolate the gain obtained from adaptive port selection. The baseline uses a fixed selected port set and does not track favorable fading peaks across the fluid aperture. It therefore provides a reference for understanding the reliability loss caused by spatial correlation.

Letting $\mathbf{\Delta}=\mathbf{x}-\hat{\mathbf{x}}$ denote the difference between two spatial-symbolic vectors, we obtain
\begin{align}
P_{\text{PEP}}^{\text{sta}}(\mathbf{x} \to \hat{\mathbf{x}}) &= \Pr \left( \|\mathbf{n}\|_2^2 > \|\mathbf{n} + \sqrt{\gamma}\mathbf{H}\mathbf{\Delta}\|_2^2 \right) \nonumber \\
&= \Pr \left( \|\mathbf{n}\|_2^2 > \|\mathbf{r}_{x}\|_2^2 \right) = \Pr(D > 0), \label{eq:pep_fas_raw}
\end{align}
where $D \triangleq \|\mathbf{n}\|_2^2-\|\mathbf{r}_{x}\|_2^2$ can be equivalently expressed as
\begin{equation}\label{eq:quadratic_form}
D = \mathbf{z}^H \mathbf{\Gamma} \mathbf{z},
\end{equation}
with the augmented vector defined as $\mathbf{z} = [\mathbf{n}^T, \mathbf{r}_x^T]^T \in \mathbb{C}^{2N_r \times 1}$ and the signature matrix $\mathbf{\Gamma} = \mathrm{diag}\{\mathbf{I}_{N_r}, -\mathbf{I}_{N_r}\}$.

Let us define $\mathbf{g} = \mathbf{H}\mathbf{\Delta} = \mathbf{H}_w \mathbf{R}^{1/2} \mathbf{\Delta} \sim \mathcal{CN}(\mathbf{0}, \mathbf{\Sigma}_g)$. Then its covariance matrix is given by
\begin{align}
\mathbf{\Sigma}_g &= \mathbb{E}[\mathbf{H}_w \mathbf{R}^{1/2} \mathbf{\Delta} \mathbf{\Delta}^H (\mathbf{R}^{1/2})^H \mathbf{H}_w^H] \triangleq \delta_x^2 \mathbf{I}_{N_r}, \label{eq:variance_g_derivation}
\end{align}
where the effective distance parameter $\delta_{x}^{2}$ is explicitly coupled with the spatial correlation coefficient $[\mathbf{R}]_{i,j}$ as
\begin{equation}\label{eq:delta_x_final}
\delta_{x}^{2} = |s_{q}|^{2} + |s_{p}|^{2} - 2\text{Re}\{s_{q}s_{p}^{*}[\mathbf{R}]_{i,j}\}.
\end{equation}

\begin{remark}
The expression~\eqref{eq:delta_x_final} explains why strong inter-port correlation degrades SM detection. When $s_q=s_p$, the effective distance becomes
$\delta_x^2=2|s_q|^2(1-[\mathbf{R}]_{i,j})$. If the two selected ports are highly correlated, i.e., $[\mathbf{R}]_{i,j}\rightarrow 1$, this distance approaches zero. The receiver then has difficulty distinguishing the two port indices. This observation supports the use of correlation-aware port selection.
\end{remark}

\begin{lemma} \label{lemma:covariance_z_fas}
Assuming $\mathbb{E}[\mathbf{n}\mathbf{n}^H]=\mathbf{I}_{N_r}$, the augmented vector $\mathbf{z}$ is circularly symmetric complex Gaussian with covariance matrix
\begin{equation}\label{eq:covariance_z}
\mathbf{\Sigma}_z = \begin{bmatrix} 1 & 1 \\ 1 & 1 + \gamma \delta_x^2 \end{bmatrix} \otimes \mathbf{I}_{N_r}.
\end{equation}
\end{lemma}

\begin{proof}
See Appendix \ref{appendix:proof_lemma2}.
\end{proof}

Based on Lemma \ref{lemma:covariance_z_fas}, the moment generating function (MGF) of $D$ is given by
\begin{equation}
\mathcal{M}_D(t) = [\det(\mathbf{I}_{2N_r} - t \mathbf{\Sigma}_z \mathbf{\Gamma})]^{-1}.
\end{equation}
The two eigenvalues of the corresponding $2\times 2$ core matrix are expressed as
\begin{equation}
\lambda_1 = \frac{-\gamma \delta_x^2 + \sqrt{\gamma^2 \delta_x^4 + 4\gamma \delta_x^2}}{2},~\lambda_2 = \frac{-\gamma \delta_x^2 - \sqrt{\gamma^2 \delta_x^4 + 4\gamma \delta_x^2}}{2}.
\end{equation}
Thus, the MGF reduces to
\begin{equation}\label{mdt2}
\mathcal{M}_D(t) = [(1 - t\lambda_1)(1 - t\lambda_2)]^{-N_r}.
\end{equation}
Using \eqref{mdt2} and \eqref{eq:pep_fas_raw}, the closed-form static baseline becomes
\begin{equation}
P_{\text{PEP}}^{\text{sta}} = \left( \frac{\lambda_{1}}{\lambda_{1}\!-\!\lambda_{2}} \right)^{N_{r}} \sum_{k=0}^{N_{r}\!-\!1} \binom{N_{r}\!+\!k\!-\!1}{k} \left( \frac{-\lambda_{2}}{\lambda_{1}\!-\!\lambda_{2}} \right)^{k}.
\label{eq:pep_exact_final_fas}
\end{equation}

\subsubsection{Average Symbol Error Rate (ASER)} \label{subsubsec:generalized_asep}
The ASER, denoted by $P_E$, is bounded by the union bound over all $MN_{\rm PS}$ joint transmission states. That is,
\begin{equation}\label{eq:generalized_union_bound}
P_E \le \frac{1}{M N_{\rm PS}} \sum_{i \in \mathcal{I}_a} \sum_{s_q \in \mathcal{S}} \sum_{j \in \mathcal{I}_a} \sum_{\substack{s_p \in \mathcal{S} \\ (j, s_p) \neq (i, s_q)}} P_{\text{PEP}}(\mathbf{x}_{i,q} \to \mathbf{x}_{j,p}).
\end{equation}
The algorithm-aware bound is obtained by substituting \eqref{eq:tight_pep} into \eqref{eq:generalized_union_bound}. The static benchmark is obtained by substituting \eqref{eq:pep_exact_final_fas}. The gap between the two results reflects the gain obtained from adaptive correlation-aware port selection.

\subsection{Asymptotic Analysis and System-Level Gains} \label{subsec:asymptotic_analysis}
We next study the high-SNR behavior of the proposed system. The asymptotic analysis identifies the diversity order and the array gain, which determine the slope and the horizontal shift of the error-rate curve, respectively. The asymptotic error probability can be expressed as
\begin{equation}
P_{asy} \simeq (G_c \gamma)^{-\mathcal{D}},
\end{equation}
where $G_c$ denotes the array gain, and $D$ denotes either the exact diversity order of the static baseline or the predicted effective diversity slope of the adaptive selection scheme.

\subsubsection{Asymptotic Error Probability and Diversity Order} \label{subsubsec:asymptotic_diversity}
For the static baseline, the decision variable depends on
\begin{equation}\label{eq:Z_decomposition}
Z=\|\mathbf{g}\|_2^2=\sum_{n=1}^{N_r}|g_n|^2\sim\Gamma(N_r,\delta_x^2).
\end{equation}
Averaging the $Q$-function over $Z$, we obtain
\begin{equation}\label{eq:craig_integral}
P^{\text{UB}}(\mathbf{x} \to \hat{\mathbf{x}}) = \frac{1}{\pi} \int_{0}^{\pi/2} \left( 1 + \frac{\gamma \delta_x^2}{4 \sin^2 \theta} \right)^{-N_r} d\theta.
\end{equation}
At high SNR, applying the asymptotic expansion yields
\begin{equation}\label{eq:asymptotic_final_fas}
P_{asy}^{\text{UB}}(\mathbf{x} \to \hat{\mathbf{x}}) = \frac{1}{2} \left( \frac{4}{\gamma \delta_x^2} \right)^{N_r} \frac{(2N_r - 1)!!}{(2N_r)!!}.
\end{equation}
The local diversity order of the conventional static baseline is evaluated by the logarithmic limit as $N_{\rm r}$.

For the proposed FAS-SM architecture, the aperture provides an additional transmit-side selection effect. Due to spatial correlation, the number of effectively independent transmit branches is limited. We denote the effective branch number by
\begin{equation} L_{\rm eff}=\min\{N,\lfloor D_s\rfloor\}.
\end{equation}
Under the equivalent-independent-branch approximation and when $N_{\rm PS}\le L_{\rm eff}$, selecting $N_{\rm PS}$ ports introduces an order-statistics penalty because the reliability of the spatial alphabet is affected by the weakest retained branch. The corresponding effective high-SNR diversity slope is predicted as \begin{equation}
d_F = (L_{\rm eff}-N_{\rm PS}+1)N_r.
\end{equation}
This expression characterizes the tradeoff between spatial-index rate and selection diversity. Increasing $N_{\rm PS}$ improves the spatial-index rate, but it reduces the available selection diversity under the effective independent-branch model.

\subsubsection{Extreme-Value Array Gain}
The diversity order determines the high-SNR slope, whereas the array gain determines the horizontal shift of the error-rate curve. To characterize the aperture-limited peak-tracking gain, we consider the scalar equivalent spatial fading process. Let $g_i\sim\mathcal{CN}(0,1)$ denote the scalar fading coefficient associated with the $i$-th effective spatial look. The scalar extreme-value array gain is defined as\label{eq:GF_def}
\begin{equation} G_F \triangleq \mathbb{E}\left[ \max_{1\le i\le N_{\rm EV}} |g_i|^2 \right].
\end{equation}
Since $|g_i|^2$ follows a unit-mean exponential distribution, the scalar equivalent extreme-value gain can be approximated as
\begin{equation} \label{eq:GF}
G_F \approx \psi(N_{\rm EV}+1)+\gamma_{\rm E},
\end{equation}
where $\psi(\cdot)$ is the Digamma function and $\gamma_{\rm E}\approx 0.5772$ is the Euler-Mascheroni constant. A derivation of this scalar equivalent expression is provided in Appendix \ref{App:Proof_ArrayGain_Detailed}. For comparison, if all $N$ candidate ports were strictly independent scalar Rayleigh fading branches, the corresponding extreme-value gain would be given by
\begin{equation} \label{eq:GF_iid}
G_F^{\rm i.i.d.}(N) = \psi(N+1)+\gamma_{\rm E} = \ln N+\gamma_{\rm E}+o(1).
\end{equation}
In a finite FAS aperture, however, spatial correlation limits the effective number of independent fading peaks from $N$ to $N_{\rm EV}$. Therefore, increasing the port density improves the array gain only until the continuous aperture is sufficiently sampled.

\section{Numerical Results and Discussion} \label{sec:numerical_results}
This section evaluates the proposed FAS-SM architecture and validates the analytical results. Unless otherwise specified, quadrature phase shift keying (QPSK) modulation is used with $M=4$, the number of receive antennas is $N_{\rm r}=4$, the number of selected ports is $N_{\rm PS}=4$, the number of candidate fluid ports is $N=300$, and the normalized aperture size is $W=2\lambda$.

Fig.~\ref{fig:complexity} compares the computational complexity of the considered schemes. The exhaustive-search benchmark has combinatorial complexity because all possible port subsets must be checked. The grouping-based method in~\cite{yang2024position} has the lowest complexity, since it uses fixed subregions and avoids channel-aware subset search. Among the proposed methods, SF-EDAS has quadratic scaling with the number of candidate ports because it computes a pairwise distance matrix. SOPS reduces the symbol-level search but requires projection operations whose cost depends on $N_{\rm r}^2$. CC-COAS achieves the most favorable scaling among the proposed schemes. It only needs channel-norm sorting and correlation checking over a finite set of thresholds. Therefore, CC-COAS is more suitable for dense FAS arrays with a large number of candidate ports.

\begin{figure}[t]
\centering
\includegraphics[width=0.85\linewidth]{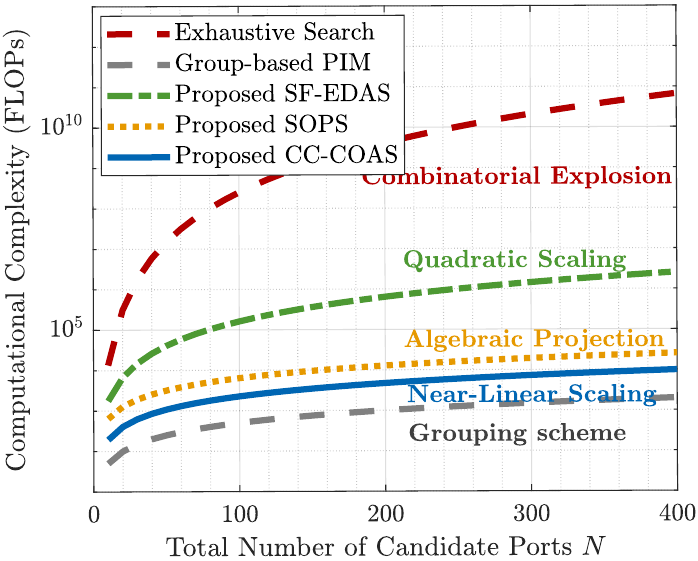}
\caption{Computational complexity in terms of floating-point operations (FLOPs) as a function of $N$.}\label{fig:complexity}
\end{figure}

\begin{figure}[t]
\centering
\subfloat[$W = \lambda$]{\includegraphics[width=0.85\columnwidth]{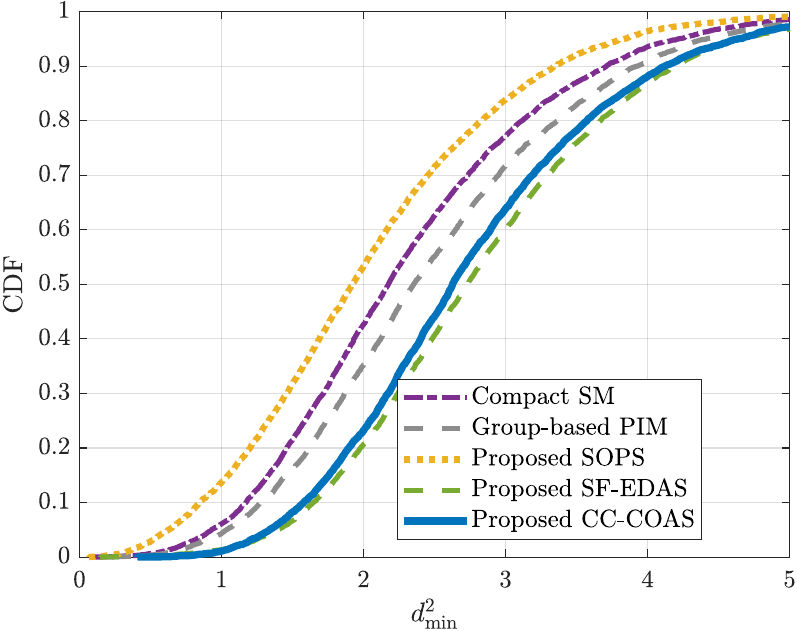}\label{fig:cdf_w1}}\hfill
\subfloat[$W = 2\lambda$]{\includegraphics[width=0.85\columnwidth]{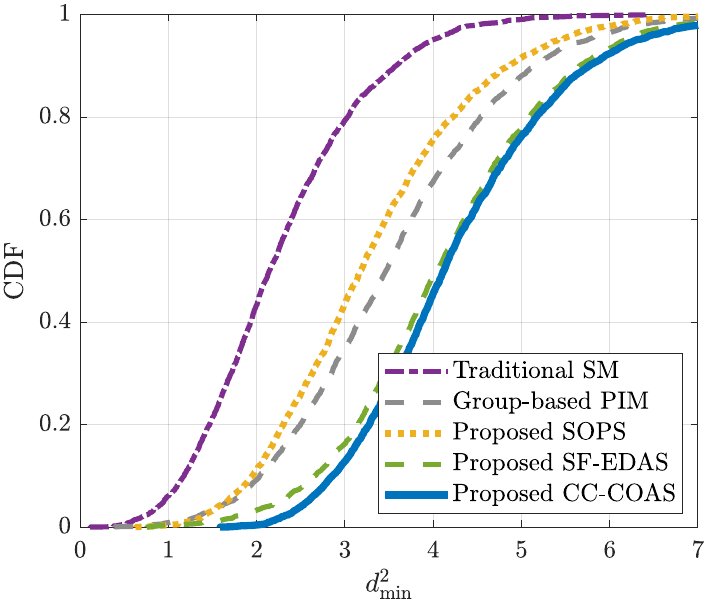}\label{fig:cdf_w2}}
\caption{Empirical CDF of the minimum squared Euclidean distance $d_{\min}^2$ evaluated under normalized spatial window sizes.}
\label{fig:cdf_comparison}
\end{figure}

In Figs.~\ref{fig:cdf_w1} and \ref{fig:cdf_w2}, we present the statistics of the minimum squared Euclidean distance to evaluate constellation separability. Under the compact aperture condition of $W=\lambda$, conventional SM exhibits the worst performance, since its fixed port configuration cannot adapt to channel fading variations. A notable observation in this tightly constrained regime is that the proposed SOPS performs statistically worse than the group-based position index modulation (PIM) scheme in \cite{yang2024position}. This is because the strict null-space projection criterion often forces the selection of ports with severely weakened channel gains to maintain orthogonality in a limited aperture. Therefore, the associated power loss outweighs the benefit of reduced spatial correlation, making this algebraic strategy less effective than the coarse grouping rule adopted in group-based PIM. When the aperture is enlarged to $W=2\lambda$, the performance ordering changes significantly, and the proposed SOPS surpasses the group-based PIM scheme. The additional spatial DoFs enable SOPS to identify nearly orthogonal ports with sufficiently strong channel gains, thereby realizing the benefit of spatial decorrelation without incurring the severe power penalty observed in the compact-aperture case.

Figs. \ref{fig:aser_w1} and \ref{fig:aser_w2} compare the ASER performance under two spatial aperture settings. Increasing the aperture from $W=\lambda$ to $W=2\lambda$ significantly improves reliability by enabling richer spatial diversity. In the strongly constrained case of $W=\lambda$, the proposed SOPS exhibits the worst performance, because its strict null-space projection criterion favors deeply faded ports to maintain orthogonality, and the associated power loss exceeds the gain from reduced correlation. With the larger aperture $W=2\lambda$, this limitation is mitigated, allowing the proposed SOPS to outperform the group-based PIM scheme in \cite{yang2024position}. In contrast, conventional SM remains the least competitive benchmark, since it fails to exploit the additional spatial DoFs. For both aperture settings, the proposed SF-EDAS and the proposed CC-COAS consistently deliver superior performance by striking a favorable balance between channel power and spatial correlation. Among them, the proposed CC-COAS closely approaches the optimal benchmark with near-linear complexity, highlighting its effectiveness in extracting spatial gains under diverse propagation conditions.

\begin{figure}[t]
\centering
\subfloat[$W = \lambda$]{\includegraphics[width=0.85\columnwidth]{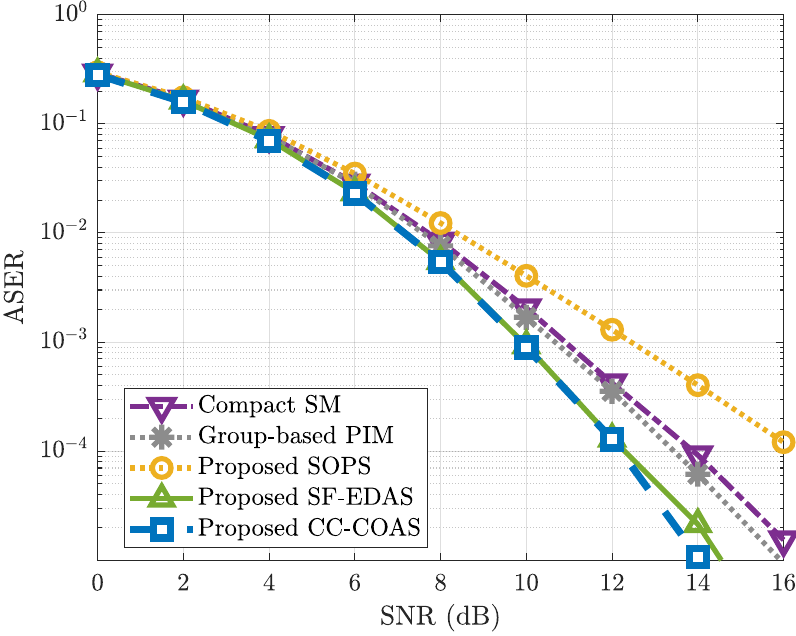}\label{fig:aser_w1}}\hfill
\subfloat[$W = 2\lambda$]{\includegraphics[width=0.85\columnwidth]{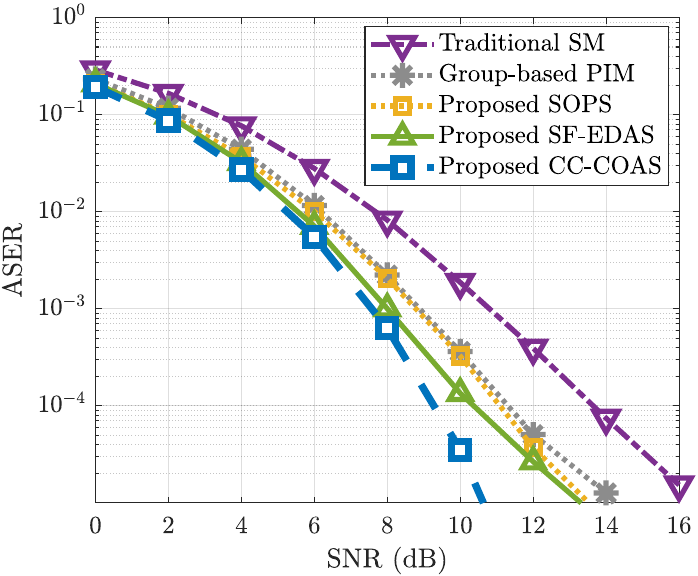}\label{fig:aser_w2}}
\caption{ASER performance evaluated under different spatial aperture constraints (SNR = 10 dB).}\label{fig:aser_comparison}
\end{figure}

\begin{figure}[t]
\centering
\subfloat[Spatial window size $W$ (SNR = 8 dB)]{\includegraphics[width=0.85\columnwidth]{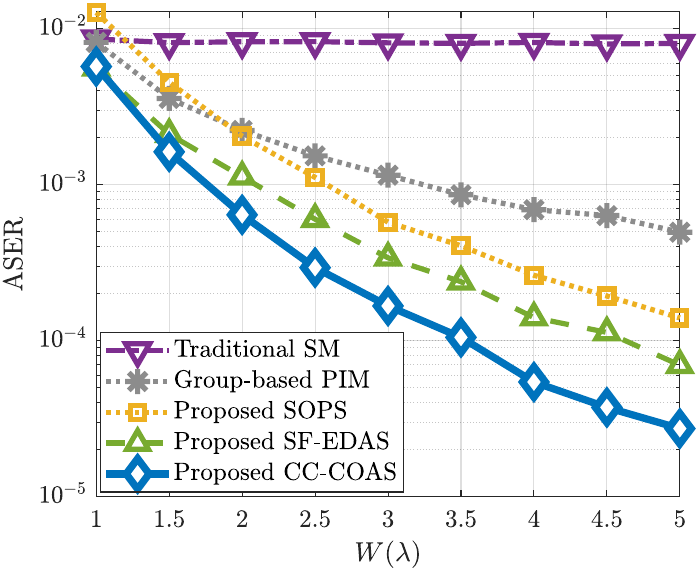}\label{fig:aser_w}}\hfill
\subfloat[Number of fluid ports $N$ (SNR = 10 dB)]{\includegraphics[width=0.85\columnwidth]{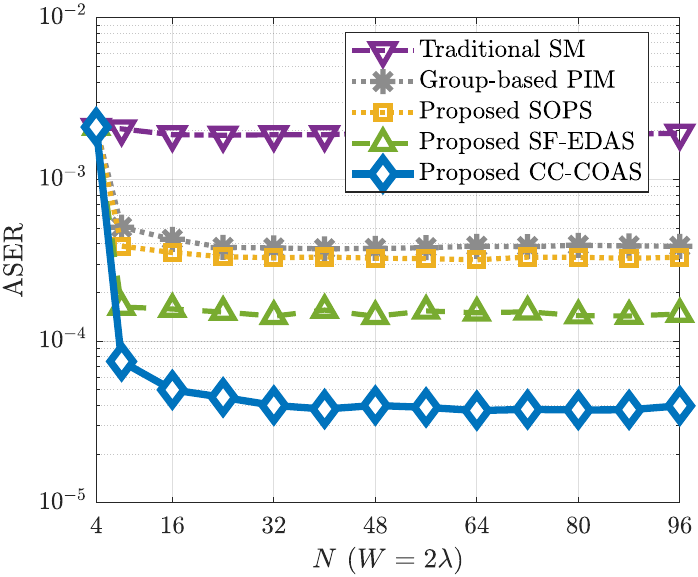}\label{fig:aser_n}}
\caption{ASER performance analysis evaluated under varying spatial aperture sizes and port densities.}\label{fig:spatial_analysis}
\end{figure}

\begin{figure}[t]
\centering
\includegraphics[width=0.85\linewidth]{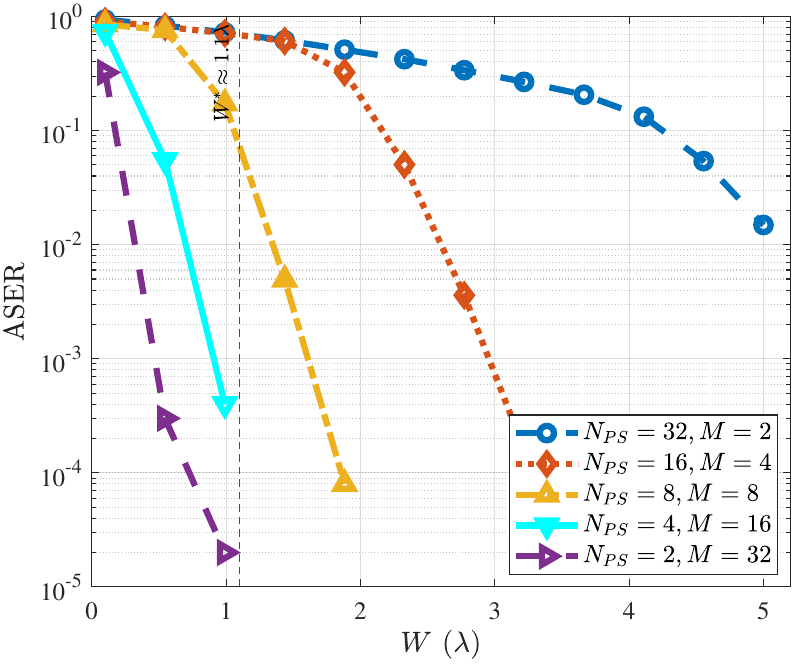}
\caption{ASER versus  $W$ for different FAS-SM configurations under $B=6$ bpcu ($\text{SNR} = 15$ dB).}\label{fig:mtrade}
\end{figure}

Figs.~\ref{fig:aser_w} and \ref{fig:aser_n} characterize the spatial diversity limits and scalability of the proposed architecture through the ASER versus the normalized window size and the total number of candidate ports. Enlarging the spatial window introduces additional spatial DoFs, which mitigates channel correlation and enhances the gain from dynamic port selection. Conventional SM remains nearly unchanged across all settings because its fixed port configuration cannot respond to instantaneous fading variations, while the group-based PIM scheme achieves only marginal improvement due to its rigid subregion partition. The proposed SOPS is ineffective under small apertures, since its strict orthogonality criterion often forces the selection of deeply faded ports. By contrast, CC-COAS achieves a clear diversity advantage by balancing spatial orthogonality and channel gain. In addition, increasing the spatial resolution enables dense fluid antenna ports to sample the continuous fading field more accurately and capture local channel peaks. Therefore, conventional SM and group-based PIM saturate early and fail to exploit the full spatial potential, whereas the proposed CC-COAS approaches the lowest theoretical bound with only a moderate number of candidate ports. This result confirms the strong scalability of the proposed CC-COAS for ultra-dense FASs without excessive computational cost.

\begin{figure*}[t]
\centering
\includegraphics[width=\textwidth]{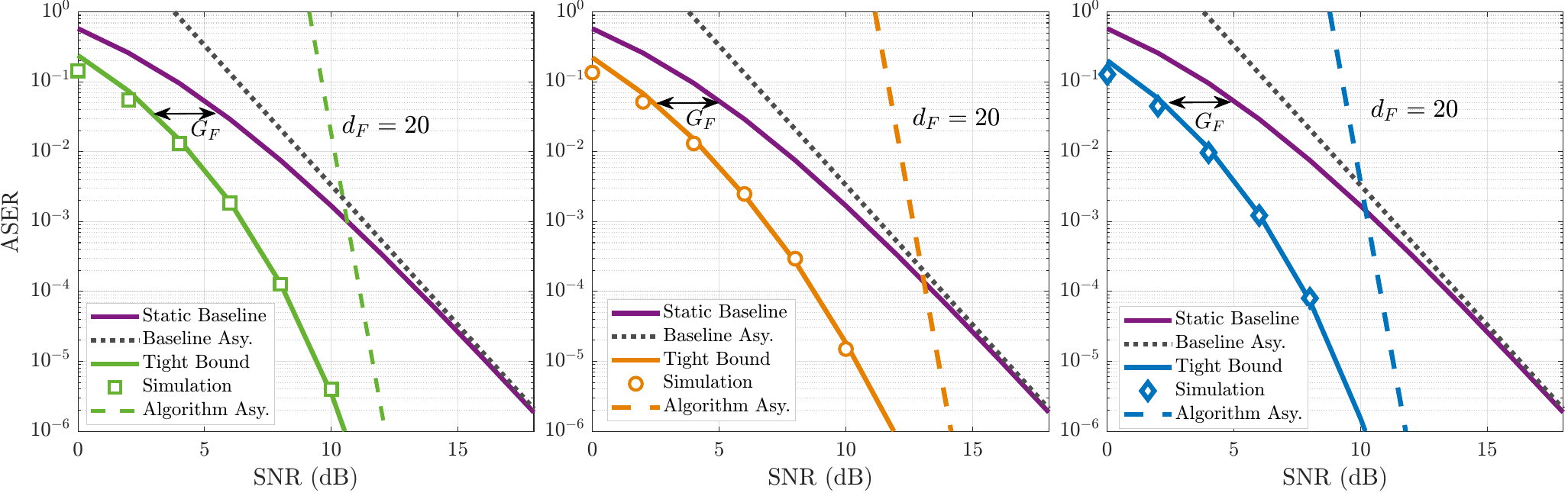}
\caption{ASER performance validation of the proposed algorithms: (a) SF-EDAS, (b) SOPS, and (c) CC-COAS.}\label{fig:performance_validation}
\end{figure*}

Fig.~\ref{fig:mtrade} illustrates the spectral-efficiency tradeoff between spatial index bits and modulation bits for the proposed Tx-SIMO-FAS SM scheme at a fixed rate of $B=6$ bpcu. In the aperture-limited regime ($W<1.1\lambda$), configurations with smaller $N_{\rm PS}$ and higher modulation orders $M$ (e.g., $N_{\rm PS}=2$, $M=32$) achieve better reliability. This is because the severely constrained aperture leads to strong spatial correlation, which degrades the separability of high-dimensional spatial constellations. In contrast, when the normalized window size exceeds the critical crossover point $W^*\approx1.1\lambda$, configurations with larger $N_{\rm PS}$ (e.g., $N_{\rm PS}=8$ and $16$) exhibit a much steeper performance improvement. This behavior indicates that a larger aperture provides richer spatial DoFs, allowing the proposed CC-COAS algorithm to exploit both extreme-value array gain and spatial diversity for more reliable index detection. Although only a single RF chain is activated, the FAS architecture effectively compensates for the reduced Euclidean distance associated with high-order modulation by leveraging the spatial agility of fluid antenna ports.

Fig.~\ref{fig:performance_validation} shows the ASER of the proposed algorithms versus the average SNR, with conventional static SM included as a benchmark. The simulation results closely match the derived analytical bounds over the entire SNR range, thereby showing good agreement with the proposed theoretical framework. In the moderate-SNR waterfall region, a clear horizontal gap can be observed between the static baseline and the proposed schemes. As indicated by the double-headed arrows, this early-stage performance gain is mainly determined by the extreme-value array gain $G_F$, which reflects the power advantage obtained by selecting the strongest fluid antenna ports before the asymptotic diversity gain becomes dominant. In the high-SNR regime, all three proposed algorithms exhibit a much steeper decay than the static benchmark. The dashed asymptotic lines confirm that the proposed schemes follow the predicted effective high-SNR slope of $d_F=20$. Although the three algorithms share the same diversity slope, slight horizontal differences remain due to their different port-selection criteria. In particular, the SOPS algorithm suffers a small power loss because of its strict successive orthogonal projection constraint, whereas both CC-COAS and SF-EDAS provide larger array gains and therefore shift the error curves further to the left, approaching near-optimal reliability.

\begin{figure}[t]
\centering
\includegraphics[width=0.85\linewidth]{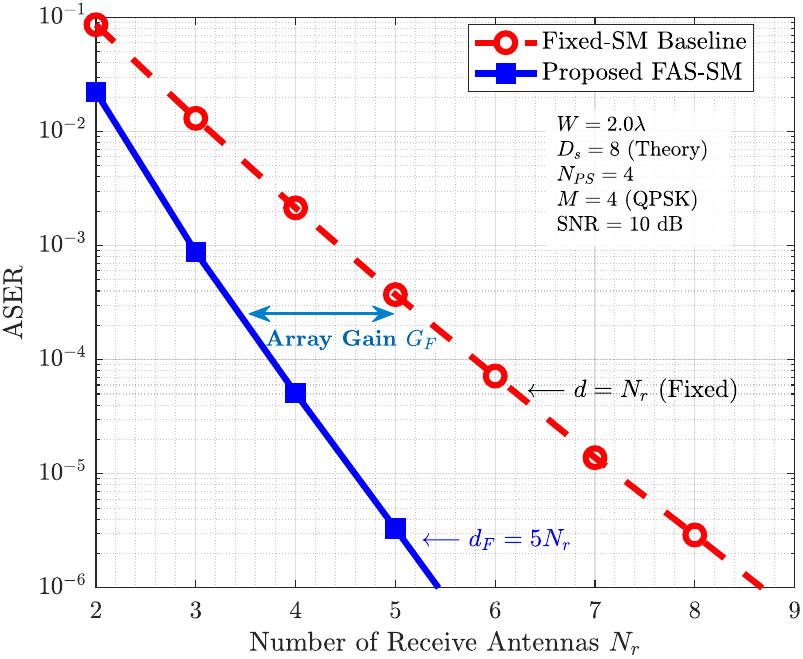}
\caption{ASER versus the number of receive antennas $N_r$ ( $N=256$).}\label{fig:Nr}
\end{figure}

In Fig.~\ref{fig:Nr}, we show the ASER as a function of the number of receive FPAs $N_r$ in order to verify the diversity scaling law of the proposed Tx-SIMO-FAS SM architecture. As can be seen, the conventional fixed-SM benchmark is fundamentally limited by the receive-side diversity and therefore maintains a diversity order of $d=N_r$. In contrast, the proposed Tx-SIMO-FAS SM scheme with the CC-COAS algorithm exhibits a much steeper decay, indicating that it can effectively exploit the spatial DoF provided by the fluid aperture. For the moderately enlarged aperture $W=2.0\lambda$, the effective energy-based DoF is identified as $D_s=8$. As a result, with $N_{\rm PS}=4$ active ports, the corresponding predicted effective diversity slope becomes $d_F=5N_r$. These results validate the multiplicative interaction between the receive-side combining gain and the transmit-side spatial selectability. In addition, the horizontal gap highlights the considerable extreme-value array gain $G_F$ enabled by the FAS-aided spatial constellation design.

In Fig.~\ref{fig:array_gain_saturation}, we illustrate the extreme-value array gain $G_F$ as a function of the number of fluid antenna ports $N$ for different normalized aperture sizes $W$. The results demonstrate that the analytical results derived from the Digamma-based expression agree closely with the simulation results, especially at the saturation plateau. In contrast to the independent and identically distributed (i.i.d.) Rayleigh benchmark, which exhibits unbounded logarithmic growth with increasing $N$, the FAS shows a clear saturation behavior, where the achievable gain is fundamentally limited by the aperture size. This observation indicates that the number of effective independent fading peaks is determined by the Rice frequency of the spatial process rather than by the sampling density itself. As the normalized aperture increases from one wavelength to six wavelengths, the saturation level rises significantly because the larger aperture introduces richer spatial fluctuations for envelope tracking. These results support the scalar equivalent approximation and confirm the aperture-limited saturation trend of FAS.

\begin{figure}[t]
\centering
\includegraphics[width=0.85\linewidth]{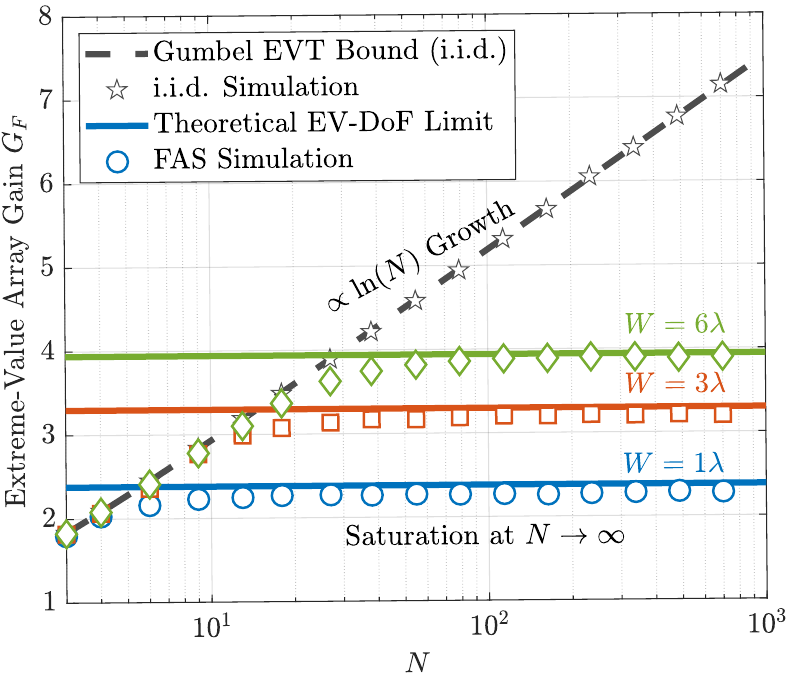}
\caption{Extreme-value array gain $G_F$ as a function of the number of fluid ports $N$.}\label{fig:array_gain_saturation}
\end{figure}

\section{Conclusion} \label{sec:conclusion}
This paper presented a systematic study of Tx-SIMO-FAS SM systems. To address the strong spatial correlation caused by limited spatial apertures, we designed effective spatial-domain constellations and proposed three port-selection algorithms for spatial-symbol mapping optimization. We also developed a theoretical framework that decouples the spatial characteristics into energy-based and extreme-value DoFs, based on which tight analytical bounds and a closed-form benchmark were derived for reliability evaluation. The theoretical analysis illustrated that the effective diversity slope follows a multiplicative scaling law with an order-statistics penalty under the equivalent independent-branch model, while the high-SNR extreme-value array gain was characterized through a scalar equivalent independent-look approximation based on the Digamma function. Simulation results showed good agreement with the analytical results, particularly in the high-SNR region. Overall, the proposed schemes outperformed the benchmark methods, although SOPS suffered noticeable degradation under highly constrained apertures because its strict orthogonality requirement led to a severe power penalty. Moreover, the results revealed that lower-order spatial indexing combined with higher-order symbol modulation is preferable in aperture-limited regimes. Finally, the observed saturation behavior further illustrated that the maximum array gain is mainly determined by the physical aperture size rather than the port sampling density. These findings establish a practical foundation for the deployment of ultra-dense FASs.

\begin{appendices}
\section{Proof of Lemma \ref{lemma:rank_preservation}}\label{appendix:proof_lemma}
We establish the pairwise non-collinearity property of the proposed CC-COAS algorithm by examining the pairwise spatial Gram matrix
\begin{equation}
\mathbf{C}_{i,j} = \begin{bmatrix}
\|\mathbf{h}_{i}\|_{2}^{2} & \mathbf{h}_{i}^H \mathbf{h}_j \\
\mathbf{h}_{j}^H \mathbf{h}_i & \|\mathbf{h}_{j}\|_{2}^{2}
\end{bmatrix}.
\end{equation}
Expanding the determinant of this Hermitian matrix yields
\begin{equation}
\det(\mathbf{C}_{i,j}) = \|\mathbf{h}_{i}\|_{2}^{2} \|\mathbf{h}_{j}\|_{2}^{2} \left( 1 - \phi(i,j)^{2} \right).
\end{equation}
Applying the correlation constraint $\phi(i,j) \le \rho_l$, we obtain a strictly positive lower bound for the determinant:
\begin{equation}
\det(\mathbf{C}_{i,j}) \ge \|\mathbf{h}_{i}\|_{2}^{2} \|\mathbf{h}_{j}\|_{2}^{2} (1 - \rho_l^{2}) > 0.
\end{equation}
For a $2 \times 2$ positive definite matrix, its minimum eigenvalue is bounded by the ratio of its determinant to its trace. Substituting the trace $\mathrm{Tr}(\mathbf{C}_{i,j}) = \|\mathbf{h}_{i}\|_{2}^{2} + \|\mathbf{h}_{j}\|_{2}^{2}$ into this fundamental inequality yields
\begin{equation}
\lambda_{\min}(\mathbf{C}_{i,j}) \ge \frac{\det(\mathbf{C}_{i,j})}{\mathrm{Tr}(\mathbf{C}_{i,j})} \ge \frac{\|\mathbf{h}_{i}\|_{2}^{2} \|\mathbf{h}_{j}\|_{2}^{2}}{\|\mathbf{h}_{i}\|_{2}^{2} + \|\mathbf{h}_{j}\|_{2}^{2}} \left( 1 - \rho_l^{2} \right) > 0.
\end{equation}
This strictly positive lower bound confirms that any selected port pair remains non-collinear under the correlation constraint. The proof of Lemma \ref{lemma:rank_preservation} is thus completed.

\section{Justification of Spatial DoF Approximation} \label{App:Proof_DoF}
To derive the energy-based DoF $D_s$, the spatial energy $E = \|\mathbf{h}\|_2^2$ is decomposed as $E = \sum_{k=1}^N \lambda_k |z_k|^2$, where $\lambda_k$ are the eigenvalues of $\mathbf{R}$. The first two moments of $E$ are
\begin{equation}
\left\{\begin{aligned}
\mathbb{E}[E] &= \sum_{k=1}^N \lambda_k,\\
\text{Var}(E) &= \sum_{k=1}^N \lambda_k^2.
\end{aligned}\right.
\end{equation}
We match these moments to a virtual system of $D_s$ independent branches with average energy $\mu$. The moments of the virtual system are defined by $\mathbb{E}[E_{\text{eq}}] = D_s \mu$ and $\text{Var}(E_{\text{eq}}) = D_s \mu^2$. By matching the mean and variance, the energy DoF is established as
\begin{equation}\label{eq:app_ds_final}
D_s = \frac{\left( \sum_{k=1}^N \lambda_k \right)^2}{\sum_{k=1}^N \lambda_k^2} = \frac{\mathrm{Tr}(\mathbf{R})^2}{\|\mathbf{R}\|_F^2}.
\end{equation}

Next, we evaluate the effective extreme-value DoF, $N_{\rm EV}$, which characterizes the peak-sampling potential of the fluid antenna aperture. For the Jakes process $R(\tau) = J_0(2\pi \tau)$, the spectral moment $\lambda_2$ is found by the curvature at the origin as $\lambda_2 = -R''(0)$. Using the identity $J_0''(0) = -1/2$, the second derivative is evaluated by
\begin{equation}
R''(0) = \left. \frac{d^2}{d\tau^2} J_0(2\pi \tau) \right|_{\tau=0} = -(2\pi)^2 \left( \frac{1}{2} \right) = -2\pi^2,
\end{equation}
which yields $\lambda_2 = 2\pi^2$. Based on the Rice frequency for Rayleigh envelopes, the density of independent maxima $\bar{n}$ is formulated by
\begin{equation}\label{eq:app_rice_detail}
\bar{n} = \sqrt{\frac{\pi}{6}} \sqrt{\lambda_2} = \sqrt{\frac{\pi}{6}} (\sqrt{2}\pi) \approx 3.21.
\end{equation}
Integrating this density over the aperture window $W$ and adding the initial sampling point gives the scaling law $N_{EV} = \lfloor \bar{n}W \rfloor + 1 \approx 3.21W + 1$. This justifies \eqref{eq:nev_result_condensed}.

\section{Proof of Lemma \ref{lemma:covariance_z_fas}}\label{appendix:proof_lemma2}
The joint covariance $\mathbf{\Sigma}_z = \mathbb{E}[\mathbf{z} \mathbf{z}^H]$ is partitioned into sub-blocks as
\begin{equation}
\mathbf{\Sigma}_z = \begin{bmatrix} \mathbb{E}[\mathbf{n}\mathbf{n}^H] & \mathbb{E}[\mathbf{n}\mathbf{r}_x^H] \\ \mathbb{E}[\mathbf{r}_x\mathbf{n}^H] & \mathbb{E}[\mathbf{r}_x\mathbf{r}_x^H] \end{bmatrix}.
\end{equation}
First, $\mathbb{E}[\mathbf{n}\mathbf{n}^H] = \mathbf{I}_{N_r}$ holds by assumption. By substituting the received signal component $\mathbf{r}_x = \mathbf{n} + \sqrt{\gamma}\mathbf{g}$, the cross-covariance is evaluated by
\begin{equation}
\mathbb{E}[\mathbf{n}\mathbf{r}_x^H] = \mathbb{E}[\mathbf{n}\mathbf{n}^H] + \sqrt{\gamma}\mathbb{E}[\mathbf{n}\mathbf{g}^H] = \mathbf{I}_{N_r},
\end{equation}
as $\mathbb{E}[\mathbf{n}\mathbf{g}^H] = \mathbf{0}$ by the statistical independence between the noise $\mathbf{n}$ and the equivalent channel $\mathbf{g}$. Similarly, $\mathbb{E}[\mathbf{r}_x\mathbf{n}^H] = \mathbf{I}_{N_r}$ holds by symmetry.

The signal auto-covariance is determined by the algebraic expansion
\begin{equation}
\mathbb{E}[\mathbf{r}_x\mathbf{r}_x^H] = \mathbb{E}[\mathbf{n}\mathbf{n}^H] + \gamma\mathbb{E}[\mathbf{g}\mathbf{g}^H] = (1 + \gamma \delta_x^2) \mathbf{I}_{N_r},
\end{equation}
as evaluated by substituting the channel covariance $\mathbb{E}[\mathbf{g}\mathbf{g}^H] = \delta_x^2 \mathbf{I}_{N_r}$. The proof is completed by assembling these blocks as
\begin{equation}
\mathbf{\Sigma}_z = \begin{bmatrix} \mathbf{I}_{N_r} & \mathbf{I}_{N_r} \\ \mathbf{I}_{N_r} & (1 + \gamma \delta_x^2)\mathbf{I}_{N_r} \end{bmatrix} = \begin{bmatrix} 1 & 1 \\ 1 & 1 + \gamma \delta_x^2 \end{bmatrix} \otimes \mathbf{I}_{N_r},
\end{equation}
thereby yielding the exact structure in \eqref{eq:covariance_z}. \hfill $\blacksquare$

\section{Derivation of the Scalar Extreme-Value Gain} \label{App:Proof_ArrayGain_Detailed}
Let $Y_1, \dots, Y_n$ be i.i.d.~variables with unit mean, where the cumulative distribution function (CDF) of each variable is $F_Y(y) = 1 - e^{-y}$. The CDF of the maximum statistic $X = \max\{Y_1, \dots, Y_n\}$ is established as
\begin{equation}\label{eq:app_X_CDF}
F_X(x) = \prod_{i=1}^n P(Y_i \le x) = (1 - e^{-x})^n.
\end{equation}
The expected value of $X$ is determined by integrating its complementary CDF over the non-negative domain
\begin{equation}\label{eq:app_E_integral}
\mathbb{E}[X] = \int_{0}^{\infty} [1 - (1 - e^{-x})^n] dx.
\end{equation}
To evaluate this integral, we perform a change of variable by setting $u = 1 - e^{-x}$, which implies $du = e^{-x} dx = (1 - u) dx$. Rearranging for $dx$ yields $dx = \frac{du}{1 - u}$. As $x$ ranges from $0$ to $\infty$, $u$ maps to the interval $[0, 1]$. Substituting these terms into \eqref{eq:app_E_integral} results in
\begin{equation}\label{eq:app_u_integral}
\mathbb{E}[X] = \int_{0}^{1} \frac{1 - u^n}{1 - u} du.
\end{equation}
Applying the finite geometric series identity $\frac{1 - u^n}{1 - u} = \sum_{k=0}^{n-1} u^k$ allows for term-wise integration
\begin{equation}\label{eq:app_harmonic_sum}
\mathbb{E}[X] = \sum_{k=0}^{n-1} \int_{0}^{1} u^k du = \sum_{k=0}^{n-1} \frac{1}{k+1} = \sum_{j=1}^n \frac{1}{j},
\end{equation}
which defines the $n$-th harmonic number $H_n$. To analytically extend this discrete result to the continuous spatial potential characterized by $N_{EV}$, we utilize the fundamental identity relating harmonic numbers to the Digamma function:
\begin{equation}\label{eq:app_digamma_identity}
H_n = \psi(n + 1) + \gamma_E.
\end{equation}
Substituting the effective independent-look number $N_{\rm EV}$ for the integer $n$ in \eqref{eq:app_digamma_identity}, we obtain the scalar equivalent approximation \begin{equation} G_F \approx \psi(N_{\rm EV}+1)+\gamma_{\rm E}. \end{equation} This result should be interpreted as an effective independent-look approximation for the aperture-limited scalar peak-tracking gain, rather than as an exact order statistic of all correlated physical ports.
\end{appendices}


\end{document}